\documentclass[a4paper,12pt]{article}%
\usepackage{amsmath}
\usepackage{amssymb}
\usepackage{amsfonts}
\usepackage[onehalfspacing,nodisplayskipstretch]{setspace}
\usepackage{graphicx}%
\providecommand{\U}[1]{\protect\rule{.1in}{.1in}}
\newtheorem{theorema}{Theorem}

\newtheorem{theorem}{Theorem}

\newtheorem{corollary}[theorem]{Corollary}

\newtheorem{lemma}[theorem]{Lemma}

\newtheorem{proposition}[theorem]{Proposition}

\newenvironment{proof}[1][Proof]{\noindent\textbf{#1.} }{\ \rule{0.5em}{0.5em}}
\def\Rev{\normalfont\textsc{Rev}}

\newcommand\ignore[1]{}
\begin{document}

\title{The Root of Revenue Continuity\thanks{Previous versions: June 2024; July 2025 (\texttt{arXiv}:2507.15735v1). We thank
Yannai Gonczarowski for his useful comments.}}
\author{Sergiu Hart\thanks{%
The Hebrew University of Jerusalem (Federmann Center for the Study of
Rationality, Department of Economics, and Department of Mathematics).\quad 
\emph{E-mail}: \texttt{hart@huji.ac.il} \quad \emph{Web site}: \texttt{%
http://www.ma.huji.ac.il/hart}} \and Noam Nisan\thanks{%
The Hebrew University of Jerusalem (Federmann Center for the Study of
Rationality and School of Computer Science and Engineering). 
Supported by a grant from the Israeli Science Foundation (ISF number 505/23).\emph{E-mail}: 
\texttt{noam@cs.huji.ac.il} \quad \emph{Web site}: \texttt{%
http://www.cs.huji.ac.il/\symbol{126}noam}}}
\maketitle

\begin{abstract}
    In the setup of selling one or more goods, 
various papers have shown, in various forms and for various purposes, 
that a small change in the distribution of a buyer's valuations may cause only a 
small change in the possible revenue that can be extracted.  
We prove a simple,
    clean, convenient, 
    and general statement to this effect:
    let $X$ and $Y$ be random valuations on $k$ additive goods, 
    let $\Rev(X)$ and $\Rev(Y)$ 
    denote their optimal revenues, 
    and let $\mathrm{W}(X,Y)$ be the Wasserstein (or ``earth mover's'') 
    distance  between them; then 
$$\left\vert \sqrt{\Rev(X)}-\sqrt{\Rev(Y)}\right\vert \leq\sqrt{\mathrm{W}(X,Y)}.$$ 
This further implies that a simple explicit modification of any optimal mechanism for $X$, namely, ``uniform discounting,'' is guaranteed to be almost optimal for any $Y$ that is close to $X$ in the Wasserstein distance.
   
\end{abstract}

\section{Introduction}

Fix a mechanism for selling some good. A small change in the buyer's value for
the good---i.e., his \textquotedblleft willingness to pay" for the good---may
result in a huge change in the seller's revenue. For example, suppose that the
seller's mechanism is to sell the good for the price of \$1. If the buyer's
value for the good is \$1.01 then the seller's revenue will be \$1, while if
the buyer's value is only \$0.99 the seller's revenue will be \$0. Thus, we
see that the revenue of a mechanism is \emph{not} a continuous function of the
buyer's value.

If, on the other hand, we are allowed to modify the mechanism once the buyer's
valuation has changed, then we can recover almost all of the revenue. In our
example, if the seller, once observing the buyer's change of value from \$1.01 to
\$0.99, is now allowed to change the price to \$0.98, then the revenue
becomes \$0.98, close to the original \$1. Thus, if we look at the revenue of
the \emph{best} mechanism, which may depend on the buyer's valuation, then it
seems to be continuous in the buyer's valuation. The issue becomes of course
much more demanding when the buyer's valuation is only \textquotedblleft
probabilistically known," i.e., as a probability distribution on valuations
(this is the seller's Bayesian belief on the buyer' valuation; formally, we
refer to it as a \textquotedblleft random valuation").

For the case of selling a single good, the above intuition may indeed be
easily formalized utilizing Myerson's characterization of (revenue-)optimal
mechanisms \cite{Mye81} (see \cite{Mon15} for such a proof). However, when
selling multiple goods things become trickier, in that optimal
mechanisms may be complex and involve multiple \textquotedblleft
menu items,\textquotedblright\ i.e., multiple allocations of goods---possibly,
fractional---and prices \cite{BCKW10, CHMS10, DW11, Tha04, MM88, MV06, DDT13, DDT15}.
Nevertheless, in many such multi-good settings the basic idea does go through
and the different prices involved in the mechanism can be simultaneously
manipulated so as to compensate for small changes in the buyer's valuation.
This basic idea has appeared in various forms in multiple papers, e.g.,
\cite{CHK07, HN13, BGN17, DW12, DLN14, MP17, GW21, BCD20}, perhaps first in
\cite{BBHY08} where it was attributed to Nisan. A clean statement to this
effect appears in \cite{HR19}, where it is shown that the revenue (of the
optimal mechanism) is a continuous function of the buyer's random valuation.

Before we proceed more formally, let us mention two common types of scenarios
where such \textquotedblleft revenue-continuity\textquotedblright\ theorems
are useful.

The first scenario is that of value \emph{approximation} or \emph{imprecision}: suppose that
the mechanism relies on some approximation of the buyer's values rather than
the exact values. This may happen either due to some limitation of access to
the valuation, or due to some rounding employed by the mechanism as part of
its algorithm, or just as a method of analysis. 
From a more abstract point of view,
there may be a multitude of reasons for various types of 
imprecision, and this is often termed ``model misspecification'' \cite{MP17}.  In such cases the input (i.e.,
the valuation distribution) changes \textquotedblleft a bit,\textquotedblright%
\ and while the revenue of any fixed mechanism may change \textquotedblleft a
lot,\textquotedblright\ revenue continuity ensures that the \emph{optimal}
revenue can only change \textquotedblleft a bit\textquotedblright\ too.

The second scenario is that of \emph{sampling}: suppose that the mechanism
does not have direct access to the valuation distribution $Y$, but only to
some finite sample of values drawn from $Y$. It is well known that the
empirical distribution on such a sample (i.e., equal probability for each
sample point) will be \textquotedblleft close,\textquotedblright\ with high
probability, to the original distribution $Y$, and thus revenue continuity
will allow us to approximate the revenue of $Y$ utilizing the sample (provided
that our sample is indeed large enough to ensure \textquotedblleft closeness"
with respect to the appropriate notion).

In order to proceed more formally, let us stick to a commonly studied setting
(see Section \ref{s:other Gamma} for alternative settings): a single seller is
selling $k$ goods to a single buyer. The buyer has a value $x_{i}$ for each
good $i$, and this simple setting assumes that the buyer's valuation for every
subset $I\subseteq\{1,...,k\}$ of the goods is $\sum_{i\in I}x_{i}$, i.e.,
additive over the goods. The buyer knows his valuation, but the seller only
knows the distribution $X$ over the buyer's valuations. We will call such an
$X$ a $k$\emph{-good} \emph{random} \emph{valuation}. Formally, think of $X$
as a random variable that takes values in $\mathbb{R}_{+}^{k}\equiv
\mathbb{R}_{\geq0}^{k}$, possibly with arbitrary correlations between the $k$
coordinates, i.e., the values of the $k$ goods. The seller must design a
mechanism for selling the goods that is individually rational (IR) and
incentive compatible (IC)---from now on, the
term \textquotedblleft mechanism"
will always mean an \textquotedblleft IR and IC mechanism"---and we denote
by$\,$\textsc{Rev}$(X)$ the expected revenue of the best mechanism for the
valuation $X.$

Thus, the desired revenue continuity theorem will state that if a $k$-good
random valuation $X$ is \textquotedblleft close\textquotedblright\ to a target
valuation $Y$ then \textsc{Rev}$(X)$ is close to \textsc{Rev}$(Y)$. There are
many different useful notions for closeness of (distributions of) random
variables (see, e.g. \cite{GS02} for a thorough comparison), and we should
obviously attempt to use as \textquotedblleft weak\textquotedblright\ a notion
of distance as possible. In particular, a theorem that uses the strong notion
of total variation distance,\footnote{Defined as the supremum of
$|\mathbb{P}[X\in A]-\mathbb{P}[Y\in A]|$ over all measurable sets
$A\subset\mathbb{R}_{+}^{k}$.} while easily proved, will be rather useless as,
for example, the two applications mentioned above (approximation and sampling)
do not yield distributions that are close in the total variation
distance.\footnote{Still \cite{BCD20} manages to indirectly use this total
variation distance.}

Hart and Reny \cite{HR19} considered the Prohorov distance $\pi$---the
standard metric for convergence in distribution---and showed that for
\emph{bounded} random valuations $X$ and $Y$ we have $|$\textsc{Rev}%
$(X)-$\textsc{Rev}$(Y)|\leq(2M+1)\sqrt{\pi(X,Y)},$ where $M\geq1$ is such
that\footnote{The convenient norm on valuations in $\mathbb{R}^{k}$ is the
$\ell_{1}$-norm, $\left\Vert x\right\Vert _{1}:=\sum_{i=1}^{k}\left\vert
x_{i}\right\vert $.} $\left\Vert
X\right\Vert _{1},\left\Vert Y\right\Vert _{1}\leq M$. This yields continuity
with respect to the Prohorov distance for uniformly bounded random valuations,
but cannot be extended to unbounded valuations (even with bounded
expectations), as can be seen in the following example from \cite{HR19}: let
$M$ be arbitrarily large, let $X$ be the 1-good random valuation that takes
the value $M$ with probability $1/M$ and value $0$ with probability $1-1/M,$
and $Y$ the constant valuation $Y\equiv0$. Then, even though $\pi(X,Y)=1/M$,
we have \textsc{Rev}$(X)=1$ and \textsc{Rev}$(Y)=0$.

We provide here a revenue continuity result that does apply to arbitrary
unbounded valuations, and is based on the \emph{Wasserstein distance}
$\mathrm{W}$, also known as the \textquotedblleft earth-mover's distance"
\cite{RTG00,Dudley02}. The distance $\mathrm{W}(X,Y)$ between two random
variables $X$ and $Y$ with values in $\mathbb{R}^{k}$ is defined as the
distance between their best \textquotedblleft coupling";\footnote{In
probability theory, a \emph{coupling} of random variables with given
distributions refers to having them defined on the same probability space;
this allows the use of pointwise comparisons.} specifically, $\mathrm{W}%
(X,Y):=\inf\{\mathbb{E}[\Vert X^{\prime}-Y^{\prime}\Vert_{1}]\},$ with the
infimum taken over all pairs of random variables $X^{\prime}$ and $Y^{\prime}$
that are defined on the same space and are such that $X^{\prime}$ has the same
distribution as $X$ and $Y^{\prime}$ has the same distribution as $Y$.

Our main result is:

\begin{theorema}
\label{th:dist(X,Y)}Let $X$ and $Y$ be $k$-good random valuations such that
the Wasserstein distance $\mathrm{W}(X,Y)$ is finite. Then either both $X$ and
$Y$ have infinite revenue, or both have finite revenue, and then
\begin{equation}
\left\vert \sqrt{\text{\normalfont\textsc{Rev}}(X)}-\sqrt{\text{\normalfont\textsc{Rev}}%
(Y)}\right\vert \leq\sqrt{\mathrm{W}(X,Y)}. \label{eq:sqrt(rev)}%
\end{equation}

\end{theorema}

The proof of this elegant (and sharp) inequality is simple; see Section
\ref{s:approx}.

As an immediate consequence of Theorem \ref{th:dist(X,Y)} we get our
convergence result:

\begin{theorema}
\label{th:W-lim}Let $X$ and $\{X_{n}\}_{n\geq1}$ be $k$-good random
valuations. If $X_{n}$ converges to $X$ in the Wasserstein distance, i.e.,
$\mathrm{W}(X_{n},X)\rightarrow0$, then\footnote{More precisely, if
\textsc{Rev}$(X)$ is finite, then so are all but finitely many of the
\textsc{Rev}$(X_{n}),$ and \textsc{Rev}$(X_{n})\rightarrow$\textsc{Rev}$(X);$
and if \textsc{Rev}$(X)$ is infinite, then so are all but finitely many of the
\textsc{Rev}$(X_{n})$.}%
\[
\lim_{n\rightarrow\infty}\text{\normalfont\textsc{Rev}}(X_{n})=\text{\normalfont\textsc{Rev}}(X).
\]

\end{theorema}

We emphasize that the random valuations are not required to be bounded or have
finite expectations. Recall that we cannot get a similar convergence result
for unbounded valuations with the Prohorov distance.\footnote{As \cite{HR19}
shows, if $X_{n}$ converges to $X$ with respect to the Prohorov distance
(i.e., in distribution), then, while we do have $\lim\inf_{n}\,$\textsc{Rev}%
$(X_{n})\geq\,$\textsc{Rev}$(X)$, in order to get $\lim_{n}\,$\textsc{Rev}%
$(X_{n})=\,$\textsc{Rev}$(X)$ we must require the sequence $X_{n}$ to be
\emph{uniformly integrable}; cf. Section \ref{s:continuity}.} Theorem
\ref{th:W-lim}, on the other hand, allows us to get revenue continuity also
for unbounded valuations, both in the approximation and in the sampling
applications suggested above.\footnote{For approximation within an additive
$\varepsilon$, the Wasserstein distance between the original valuation and the
approximating one is always bounded by $\varepsilon$; for a multiplicative
approximation (e.g., to a nearest integer power of $(1+\varepsilon)$), finite
expectation of the original valuation suffices; for sampling, a finite higher
moment is required from the original valuation (see, e.g., \cite{FG15}).} See
Section \ref{s:continuity}.

As demonstrated in the first paragraph of the paper, we cannot just expect
that a \textquotedblleft good\textquotedblright\ mechanism for a random
valuation $X$ will yield high revenue also for a \textquotedblleft
close\textquotedblright\ random valuation $Y$, but rather we may need to
modify the original mechanism in order to get good revenue from $Y$ as well.
This modification is often done using what is called \textquotedblleft nudge
and round,\textquotedblright\ where the menu is discretized after some
decrease in prices.  Like some of the previous results in this vein, but unlike most others, our modification turns out to require only the discounting of all
prices by some small uniform factor, with no discretization of allocations
needed.\footnote{For a simple illustration, consider the elementary example at
the beginning of the paper. Assume that the value of the buyer changes, from
\$1.01, either up or down by no more than \$0.02, and the seller does not know which. \emph{Discounting} the
price of \$1.00 by 2\% to \$0.98 \emph{guarantees}, for all such changes, that
the revenue cannot go down by more than \$0.02.} One advantage of this
simple modification is that the discounted mechanism is explicitly given by a straightforward
transformation of the given mechanism (whereas it is not explicitly known when
we change the allocations as well). A second advantage is that the revenue
continuity result holds also for the optimal revenue that can be extracted
when restricted to various specific subclasses of mechanisms---e.g., monotonic
mechanisms \cite{HR15}, buy-many mechanisms \cite{CTT19}, or bounded menu-size mechanisms \cite{HN13}---as long as the
class is closed under such uniform discounting (which all standard classes of
mechanisms satisfy).

Specifically, for a mechanism $\mu$ and a $k$-good random valuation $X$, denote
by $R(\mu;X)$ the revenue that $\mu$ extracts from $X$. For a
\textquotedblleft discount rate\textquotedblright\ $0<\eta<1,$ the resulting
$\eta$-\emph{discounted} mechanism, denoted by $\mu_{1-\eta}$, is the
mechanism where the prices of all menu entries are discounted by $\eta,$ i.e.,
multiplied by the factor\footnote{The discounted mechanism is obtained by an
extremely simple transformation of the original mechanism (see
(\ref{eq:mu_(1-eta)}) in Section \ref{s:rescale}): for each valuation $x$, the allocation of goods that $x$ gets 
in $\mu_{1-\eta}$ is the same as the allocation of goods that the valuation
$x/(1-\eta)$ gets in $\mu,$ and the payment of $x$ in $\mu_{1-\eta}$ is $(1-\eta)$
times the payment of $x/(1-\eta)$ in $\mu.$} $1-\eta$. We show the following:

\begin{theorema}
\label{th:C}For every $\varepsilon,M>0$ with $\varepsilon<2M$, there exist
$\delta>0$ and $0<\eta<1$ such that%
\begin{equation}
R(\mu_{1-\eta};Y)\geq R(\mu;X)-\varepsilon\label{eq:eps}%
\end{equation}
for any two $k$-good random valuations $X$ and $Y$ whose Wasserstein distance
is at most $\delta,$ and for any mechanism $\mu$ such that $R(\mu;X)\leq M;$
moreover, if $\mu$ is optimal for $X$ then $\mu_{1-\eta}$ is $2\varepsilon
$-optimal for\footnote{I.e., $R(\mu_{1-\eta};Y)\geq\,$\textsc{Rev}%
$(Y)-2\varepsilon.$ As $Y$ may have a slightly higher revenue than $X$, we may
lose $2\varepsilon$ relative to the optimal mechanism of $Y,$ even though we
only lose $\varepsilon$ relative to the optimal mechanism of $X$.} $Y.$ In particular, this
holds for $\delta=\varepsilon^{2}/(4M)$ and $\eta=\varepsilon/(2M).$
\end{theorema}

This says that if the Wasserstein distance between $X$ and $Y$ is at most
$\delta,$ and a mechanism $\mu$ extracts an amount $R$ from $X$, then its
$\eta$-discounted version $\mu_{1-\eta}$ is guaranteed to extract at least
$R-\varepsilon$ from $Y;$ here $\delta$ and $\eta$ depend only on
$\varepsilon$ and a bound $M$ on the revenue $R.$

Theorem \ref{th:C} immediately and explicitly gives us the types of learning
applications that we desire: in order to find a mechanism that extracts almost
all the optimal revenue from an unknown random valuation $Y$, all we need to
do is take enough random samples $y_{1},...,y_{n}$ from $Y$ so that the
uniform distribution $X$ over the sample $\{y_{1},...,y_{n}\}$ is, with high
probability, $\delta$-close in the Wasserstein distance to $Y$ (see, e.g.,
\cite{FG15}, for conditions under which this can be done), and then take the
optimal mechanism $\mu$ for $X$ (which can be found by solving a linear
programming problem). Theorem \ref{th:C} ensures\footnote{See Remark (a)
in Section \ref{s:rescale} on the various ways to obtain the bound $M.$} that
the discounted mechanism $\mu_{1-\eta}$ is close to being optimal on $Y$.
Somewhat surprisingly, this holds despite the fact that $\mu$ may be an
extreme \textquotedblleft overfit" for $X$ (as it will be when solving the
linear program). This \textquotedblleft modular\textquotedblright%
\ decomposition of the proof of learnability is conceptually similar yet
simpler and more general than that obtained by \cite{BCD20} who require
(randomized) rounding in order to get convergence in total variation distance
(but their result, unlike ours, also holds, approximately, for multiple buyers).

While our discussion has focused on $k$-good valuations that are additive over
the goods, the approach and the results extend to general spaces of
allocations and valuations, and thus to models such as unit demand and general implementation setups, and also when restricted to deterministic mechanisms. One
just needs to replace the norm that underlies the Wasserstein distance (in our
case, the $\ell_{1}$-norm) with a (semi)norm that fits the corresponding
setup. See Section \ref{s:other Gamma}.

\section{The Basic Model\label{s:model}}

The notation is standard, and follows \cite{HN12,HN13,HR15,HR19}. For
generalizations of the basic additive model, see Section \ref{s:other Gamma}.

One seller (a monopolist) is selling a number $k\geq1$ of indivisible
goods\ (or items, objects, etc.) to one buyer; let $K:=\{1,2,...,k\}$ denote
the set of goods. The goods have no cost or value to the seller, and their
values to the buyer are $x_{1},x_{2},...,x_{k}\geq0$ (one may view the buyer's
\textquotedblleft value" of good $i$ as his \textquotedblleft willingness to
pay" for good $i$). The valuation of getting a set of goods is \emph{additive}%
: each subset $I\subseteq K$ of goods is worth $x(I):=\sum_{i\in I}x_{i}$ to
the buyer (and so, in particular, the buyer's demand is not restricted to one
good only). The valuation of the goods is given by a random variable
$X=(X_{1},X_{2},...,X_{k})$ that takes values in\footnote{For vectors
$x=(x_{1},x_{2},...,x_{k})$ in $\mathbb{R}^{k},$ we write $x\geq0$ when
$x_{i}\geq0$ for all $i$. The
nonnegative orthant is $\mathbb{R}_{+}^{k}=\{x\in\mathbb{R}^{k}:x\geq0\},$ and
$x\cdot y=\sum_{i=1}^{k}x_{i}y_{i}$ is the scalar product of $x$ and $y$ in
$\mathbb{R}^{k}.$} $\mathbb{R}_{+}^{k}$; we will refer to $X$ as a
$k$\emph{-good} \emph{random valuation.}\footnote{Of course, only the
distribution of the random valuation $X$ matters for the revenue. As in
probability theory (e.g., $\mathbb{E}\left[  X\right]  $), working with random
variables rather than their distributions has the advantage that one can easily apply
pointwise operations, such as changing each value by at most $\delta$.} The
realization $x=(x_{1},x_{2},...,x_{k})\in\mathbb{R}_{+}^{k}$ of $X$ is known
to the buyer, but not to the seller, who knows only the distribution of $X$
(which may be viewed as the seller's belief); we refer to a buyer with value
$x$ also as a buyer of \emph{type }$x$. The buyer and the seller are assumed
to be risk neutral and have quasilinear utilities.

The objective is to \emph{maximize} the seller's (expected) \emph{revenue}.

As was well established by the so-called Revelation Principle\ (starting with
\cite{Mye81}; see \cite{Krishna02}), we can restrict ourselves to
\textquotedblleft direct mechanisms" and \textquotedblleft truthful
equilibria.\textquotedblright\ A direct\emph{\ mechanism} $\mu$ consists of a
pair of functions\footnote{All functions are assumed to be Borel-measurable;
see \cite{HR15} (footnotes 10 and 48).} $(q,s)$, where $q=(q_{1},q_{2}%
,...,q_{k}):\mathbb{R}_{+}^{k}\rightarrow\lbrack0,1]^{k}$ and $s:\mathbb{R}%
_{+}^{k}\rightarrow\mathbb{R}$, which prescribe the \emph{allocation} and the
\emph{payment}, respectively. Specifically, if the buyer reports a value
vector $x\in\mathbb{R}_{+}^{k}$, then $q_{i}(x)\in\lbrack0,1]$ is the
probability that the buyer receives good\footnote{An alternative
interpretation of the model is that the goods are infinitely divisible and the
valuation is linear in quantity, in which case $q_{i}$ is the quantity (i.e.,
fraction) of good $i$ that the buyer gets.} $i$ (for each $i=1,2,...,k$), and
$s(x)$ is the payment that the seller receives from the buyer. When the buyer
of type $x$ reports his type truthfully, his payoff is $b(x)=\sum_{i=1}%
^{k}q_{i}(x)x_{i}-s(x)=q(x)\cdot x-s(x)$, and the seller's payoff is $s(x).$

The mechanism $\mu=(q,s)$ satisfies \emph{individual rationality}
(IR\textbf{)} if $b(x)\geq0$ for every $x\in\mathbb{R}_{+}^{k}$; it satisfies
\emph{incentive compatibility} (IC) if $b(x)\geq q(y)\cdot x-s(y)$ for every
alternative report $y\in\mathbb{R}_{+}^{k}$ of the buyer when his value is
$x$, for every $x\in\mathbb{R}_{+}^{k}$. From now on\emph{\ we will assume
that all mechanisms }$\mu$\emph{\ are }$k$\emph{-good mechanisms that are
given in direct form, }$\mu=(q,s):\mathbb{R}_{+}^{k}\rightarrow\lbrack
0,1]^{k}\times\mathbb{R}_{+}$, \emph{and satisfy IC and IR}.\footnote{The IR
condition is not needed for some of the basic results, such as Lemma
\ref{l:lambda} below.}

The (expected) revenue of a mechanism $\mu=(q,s)$ from a buyer with valuation
$X$, which we denote by $R(\mu;X)$, is the expectation of the payment received
by the seller, that is, $R(\mu;X)=\mathbb{E}\left[  s(X)\right]  $. We now define

\begin{itemize}
\item \textsc{Rev}$(X)$, the \emph{optimal revenue}, is the maximal revenue
that can be obtained: \textsc{Rev}$(X):=\sup_{\mu}R(\mu;X)$, where the supremum
is taken over the class of all mechanisms $\mu.$
\end{itemize}

When there is only one good, i.e., when $k=1$, we have \cite{Mye81}
\begin{equation}
\text{\textsc{Rev}}(X)=\sup_{t\geq0}t\cdot(1-F(t)),\label{eq:one good}%
\end{equation}
where $F(t):=\mathbb{P}\left[  X\leq t\right]  $ is the cumulative
distribution function of $X$. Thus, it is optimal for the seller to
\textquotedblleft post" a price $p$ (a maximizer of (\ref{eq:one good})), and
then the buyer buys the good for the price $p$ whenever his value is at least
$p$; in other words, the seller makes the buyer a \textquotedblleft
take-it-or-leave-it" offer to buy the good at price $p.$

Besides the maximal revenue \textsc{Rev}$(X)$, we are also interested in what
can be obtained from certain classes of mechanisms. For any class
$\mathcal{N}$ of mechanisms we denote

\begin{itemize}
\item $\mathcal{N}$-\textsc{Rev}$(X):=\sup_{\mu\in\mathcal{N}}R(\mu;X)$, the
maximal revenue that can be obtained by mechanisms in the class $\mathcal{N}.$
\end{itemize}

Interesting classes of mechanisms are: selling separately, bundled, deterministic,
menu size $m,$ monotonic, allocation monotonic, submodular, supermodular,
subadditive, superadditive, buy-many.

A $k$-good random valuation $X$ is \emph{integrable} if it has finite
expectation: $\mathbb{E}\left[  \left\Vert X\right\Vert \right]  <\infty;$
i.e., $\mathbb{E}\left[  X_{i}\right]  <\infty$ for every $i=1,...,k.$ It is
convenient to use on $\mathbb{R}^{k}$ the $\ell_{1}$\emph{-norm} $\left\Vert
z\right\Vert _{1}=\sum_{i=1}^{k}\left\vert z_{i}\right\vert $; for instance, since $s(x)\leq q(x)\cdot x\leq\left\Vert
x\right\Vert _{1}$ (the first inequality is by IR, the second because
$q\leq(1,1,...,1)$), we have \textsc{Rev}$(X)\leq\mathbb{E}\left[  \left\Vert
X\right\Vert _{1}\right]  $ (and this is tight for constant valuations).

Finally, the \emph{Wasserstein distance}\footnote{Also known as the
\textquotedblleft Kantorovich--Rubinstein metric," and the \textquotedblleft
earth mover's distance" (EMD). See \cite{Dudley02} (Section 11.8)\emph{.}}
$\mathrm{W}(X,Y)$ between two $k$-good random valuations $X$ and $Y$ (more
precisely, between their distributions) is defined as
\begin{align*}
\mathrm{W}(X,Y):= &  \inf\,\{\,\mathbb{E}\left[  \left\Vert X^{\prime
}-Y^{\prime}\right\Vert _{1}\right]  \,:\,X^{\prime},Y^{\prime}\text{ are
defined on the same space,}\\
&  ~\;\;\;\hspace{1.3in}X^{\prime}\overset{\mathcal{D}}{\sim}X\text{ and
}Y^{\prime}\overset{\mathcal{D}}{\sim}Y\},
\end{align*}
where $Z^{\prime}\overset{\mathcal{D}}{\sim}Z$ means that $Z^{\prime}$ and $Z$
have the same distribution (see Sections \ref{s:continuity} and
\ref{s:other Gamma} for alternative distances).  As we do not assume that the
random valuations are integrable, the Wasserstein distance may well be infinite.

\section{Revenue Approximation\label{s:approx}}

We prove here Theorem \ref{th:dist(X,Y)}, which is easily seen to be
equivalent to the following two inequalities always holding, regardless of
whether any of the terms $\mathrm{W}(X,Y),$ \textsc{Rev}$(X),$ and
\textsc{Rev}$(Y)$ are finite or infinite:\footnote{Indeed, the inequalities
matter only when $\mathrm{W}(X,Y)$ is finite, and then, if one of the revenues,
say$\,$\textsc{Rev}$(X)$, is finite, then (\ref{eq:XY}) implies that so is the
other revenue \textsc{Rev}$(Y).$}%
\begin{align}
\sqrt{\text{\textsc{Rev}}(X)} &  \leq\sqrt{\text{\textsc{Rev}}(Y)}%
+\sqrt{\mathrm{W}(X,Y)}\text{\ \ \ \ and}\label{eq:YX}\\
\sqrt{\text{\textsc{Rev}}(Y)} &  \leq\sqrt{\text{\textsc{Rev}}(X)}%
+\sqrt{\mathrm{W}(X,Y)}.\label{eq:XY}%
\end{align}
Proposition \ref{p:rev-sqrt} below proves (\ref{eq:YX}); interchanging $X$ and
$Y$ then yields (\ref{eq:XY}).

We start with a simple lemma that provides the basic inequality that is at the
root of our proof.

\begin{lemma}
\label{l:lambda}Let $\mu=(q,s)$ be an IC mechanism. Then%
\[
s(x)\leq s(\lambda y)+\frac{\lambda}{\lambda-1}\left\Vert x-y\right\Vert _{1}%
\]
for every $x,y\in\mathbb{R}_{+}^{k}$ and $\lambda>1.$
\end{lemma}

\begin{proof}
Consider the two IC inequalities%
\begin{align*}
q(x)\cdot x-s(x)  &  \geq q(\lambda y)\cdot x-s(\lambda y),\\
q(\lambda y)\cdot\lambda y-s(\lambda y)  &  \geq q(x)\cdot\lambda y-s(x).
\end{align*}
Divide the second one by $\lambda$ and add it to the first one to get
\[
(q(x)-q(\lambda y))\cdot(x-y)\geq\left(  1-\frac{1}{\lambda}\right)
(s(x)-s(\lambda y)).
\]
Now each coordinate of $q(x)-q(\lambda y)$ is between $-1$ and $1$ (because
$q(x)$ and $q(\lambda y)$ are in $[0,1]^{k}$), and so the lefthand side is at
most $\sum_{i=1}^{k}\left\vert x_{i}-y_{i}\right\vert =\left\Vert
x-y\right\Vert _{1};$ dividing by $1-1/\lambda>0$ yields the result.
\end{proof}

\bigskip

Applying this to valuations yields:

\begin{lemma}
\label{l:rev-lambda}Let $X$ and $Y$ be two random valuations, and let
$\lambda>1.$ Then%
\begin{equation}
\text{\normalfont\textsc{Rev}}(X)\leq\lambda\text{\normalfont\textsc{Rev}}(Y)+\frac{\lambda}%
{\lambda-1}\mathrm{W}(X,Y). \label{eq:lambda}%
\end{equation}

\end{lemma}

\begin{proof}
Assume that \textsc{Rev}$(Y)$ and $\mathrm{W}(X,Y)$ are both finite (otherwise
there is nothing to prove), and take $X$ and $Y$ to be defined on the same
space with $\mathbb{E}\left[  \left\Vert X-Y\right\Vert _{1}\right]  $ finite.
For every IC mechanism $\mu=(q,s)$ the inequality%
\[
s(X)\leq s(\lambda Y)+\frac{\lambda}{\lambda-1}\left\Vert X-Y\right\Vert _{1}%
\]
holds everywhere by Lemma \ref{l:lambda}. Taking expectation yields%
\[
R(\mu;X)\leq R(\mu;\lambda Y)+\frac{\lambda}{\lambda-1}\mathbb{E}\left[
\left\Vert X-Y\right\Vert _{1}\right]  ,
\]
and then the infimum over all couplings of $X$ and $Y$ gives
\begin{equation}
R(\mu;X)\leq R(\mu;\lambda Y)+\frac{\lambda}{\lambda-1}\mathrm{W}%
(X,Y).\label{eq:R-R}%
\end{equation}
Therefore, taking the supremum over all $\mu$, we get%
\[
\text{\textsc{Rev}}(X)\leq\text{\textsc{Rev}}(\lambda Y)+\frac{\lambda
}{\lambda-1}\mathrm{W}(X,Y)=\lambda\,\text{\textsc{Rev}}(Y)+\frac{\lambda
}{\lambda-1}\mathrm{W}(X,Y)
\]
(for the obvious equality \textsc{Rev}$(\lambda Y)=\lambda\,$\textsc{Rev}$(Y)$
see Proposition \ref{p:rescale} below).
\end{proof}

\bigskip

Optimizing on $\lambda$ then yields

\begin{proposition}
\label{p:rev-sqrt}Let $X$ and $Y$ be two random valuations. Then%
\begin{equation}
\sqrt{\text{\normalfont\textsc{Rev}}(X)}\leq\sqrt{\text{\normalfont\textsc{Rev}}(Y)}+\sqrt
{\mathrm{W}(X,Y)}. \label{eq:sqrt}%
\end{equation}

\end{proposition}

\begin{proof}
Assume again that \textsc{Rev}$(Y)$ and $\mathrm{W}(X,Y)$ are both finite
(otherwise, there is nothing to prove). Taking the infimum over all
$\lambda>1$ of (\ref{eq:lambda}) yields%
\[
\text{\textsc{Rev}}(X)\leq\left(  \sqrt{\text{\textsc{Rev}}(Y)}+\sqrt
{\mathrm{W}(X,Y)}\right)  ^{2}.
\]
Indeed, this is immediate when \textsc{Rev}$(Y)$ or $\mathrm{W}(X,Y)$
vanishes; and, when they are both positive, this follows because the infimum is attained at
$\lambda=1+\sqrt{\mathrm{W}(X,Y)/\text{\textsc{Rev}}(Y)}$.
\end{proof}

\bigskip

As stated at the beginning of the section, Proposition \ref{p:rev-sqrt} proves
Theorem \ref{th:dist(X,Y)}.

\bigskip

Now inequality (\ref{eq:sqrt(rev)}) does not yield a finite bound on the
difference\linebreak \textsc{Rev}$(X)-$\textsc{Rev}$(Y)$ between the two revenues;
indeed, this difference may be arbitrarily large while $\sqrt
{\text{\textsc{Rev}}(X)}-\sqrt{\text{\textsc{Rev}}(Y)}$ is arbitrarily small.
For example, for large $M$ let $X$ and $Y$ be the one-good constant valuations $X\equiv
M^{4}+2M$ and $Y\equiv M^{4}$; then \textsc{Rev}$(X)-$%
\textsc{Rev}$(Y)=X-Y=2M,$ while $\sqrt{\text{\textsc{Rev}}(X)}-\sqrt
{\text{\textsc{Rev}}(Y)}=\sqrt{M^{4}+2M}-\sqrt{M^{4}}<1/M$. Moreover, there is
no finite constant $C$ such that $\left\vert \text{\textsc{Rev}}%
(X)-\text{\textsc{Rev}}(Y)\right\vert \leq C\cdot\mathrm{W}(X,Y)$ for all
$X,Y$ (even if limited to bounded $X,Y$). Indeed, for one good, let
$X\equiv1,$ and let $Y$ take the value $1-\varepsilon$ with probability
$\varepsilon$ and $1$ with probability $1-\varepsilon.$ Then \textsc{Rev}%
$(X)=1$ and \textsc{Rev}$(Y)=1-\varepsilon$ (both $1-\varepsilon$ and $1$ are
optimal prices for $Y$), and so \textsc{Rev}$(X)-$\textsc{Rev}$(Y)=\varepsilon
,$ whereas $\mathrm{W}(X,Y)=\mathbb{E}\left[  1-Y\right]  =\varepsilon^{2}$.

However, given a bound on (at least one of) the revenues (such as the 
expectation of the random valuation), we get the following:

\begin{corollary}
\label{c:M}Let $X$ and $Y$ be $k$-good random valuations such that the
Wasserstein distance $\mathrm{W}(X,Y)$ is finite. If one of the two revenues
is bounded by $M$ (i.e., $\min\{$\normalfont\textsc{Rev}$(X),\,$\normalfont\textsc{Rev}$(Y)\}\leq
M$), then%
\begin{equation}
\left\vert \text{\textsc{Rev}}(X)-\text{\textsc{Rev}}(Y)\right\vert \leq
2\sqrt{M\cdot\mathrm{W}(X,Y)}+\mathrm{W}(X,Y).\label{eq:sqrt(d+)}%
\end{equation}
Moreover, if both revenues are bounded by $M$ (i.e., $\max\{$\textsc{Rev}%
$(X),\,$\textsc{Rev}$(Y)\}\leq M$) and\footnote{When both revenues are bounded
by $M$ we trivially have $\left\vert \sqrt{\text{\textsc{Rev}}(X)}%
-\sqrt{\text{\textsc{Rev}}(Y)}\right\vert \leq\sqrt{M},$ and so a Wasserstein
distance $\mathrm{W}(X,Y)$ that exceeds $M$ does not provide any useful bound
in (\ref{eq:sqrt(rev)}); one may thus replace $\mathrm{W}(X,Y)$ with
$\min\{\mathrm{W}(X,Y),M\}$ in this case.} $\mathrm{W}(X,Y)\leq M$, then%
\begin{equation}
\left\vert \text{\textsc{Rev}}(X)-\text{\textsc{Rev}}(Y)\right\vert \leq
2\sqrt{M\cdot\mathrm{W}(X,Y)}-\mathrm{W}(X,Y).\label{eq:sqrt(d)}%
\end{equation}

\end{corollary}

\begin{proof}
Let $r_{1}:=$ \textsc{Rev}$(X)$, $r_{2}:=\,$\textsc{Rev}$(Y)$ and
$w:=\mathrm{W}(X,Y).$ Assume w.l.o.g. that $r_{1}\leq r_{2}.$ For
(\ref{eq:sqrt(d+)}), square the inequality (\ref{eq:XY}), i.e., $\sqrt{r_{2}%
}\leq\sqrt{r_{1}}+\sqrt{w},$ and use $r_{1}\leq M.$ For (\ref{eq:sqrt(d)}),
if $r_{2}\geq w$ then square the inequality $\sqrt{r_{1}}\geq\sqrt{r_{2}%
}-\sqrt{w}\geq0$ and use $r_{2}\leq M;$ and if $r_{2}\leq w$ then $\left\vert
r_{1}-r_{2}\right\vert \leq r_{2}\leq w\leq2\sqrt{M\cdot w}-w,$ because $w\leq
M.$
\end{proof}

\bigskip

The inequalities (\ref{eq:sqrt(rev)}), (\ref{eq:sqrt(d+)}) (with $M=\min
\{$\textsc{Rev}$(X),\,$\textsc{Rev}$(Y)\}$), and (\ref{eq:sqrt(d)}) (with
$M=\max\{\text{\textsc{Rev}}(X),\,\text{\textsc{Rev}}(Y)\}$), are all sharp:
take, in the one-good case, $X\equiv c>0$ and $Y\equiv0.$

\bigskip

We now compare the result of Theorem \ref{th:dist(X,Y)} with that of Hart and
Reny \cite{HR19} (Proposition 11 in Appendix A). Recall that the
\emph{Prohorov distance} $\pi(X,Y)$ between (the distributions of) $X$ and $Y$
is given by $\pi(X,Y):=\inf\{\varepsilon>0:\mathbb{P}\left[  X\in A\right]
\leq\mathbb{P}\left[  Y\in A^{\varepsilon}\right]  +\varepsilon,$
$\mathbb{P}\left[  Y\in A\right]  \leq\mathbb{P}\left[  X\in A^{\varepsilon
}\right]  +\varepsilon$ for every Borel set $A\},$ where $A^{\varepsilon
}:=\{x:\,$there is $y\in A$ with $\left\Vert x-y\right\Vert _{1}%
\leq\varepsilon\}$ denotes the $\varepsilon$-neighborhood of the set $A.$ The
Wasserstein distance is stronger than the Prohorov distance, and the two are
equivalent for bounded random variables; specifically,\footnote{The two
inequalities are easily obtained from the \textquotedblleft coupling" result
of Strassen \cite{Str65} (Corollary of Theorem 11); see also
\cite{Billingsley99} (Theorem 6.9) and \cite{GS02}.}
\begin{equation}
\pi(X,Y)\leq\sqrt{\mathrm{W}(X,Y)},\label{eq:pi-le-sqrt(w)}%
\end{equation}
and
\[
\text{if\  }\left\Vert X\right\Vert _{1},\left\Vert Y\right\Vert _{1}\leq
M\text{ \ \ then\ \  }\mathrm{W}(X,Y)\leq(M+1)\pi(X,Y).
\]
Corollary \ref{c:M} (see (\ref{eq:sqrt(d)}) then yields
\[
\left\vert
\text{\textsc{Rev}}(X)-\text{\textsc{Rev}}(Y)\right\vert \leq2\sqrt
{M\cdot(M+1)\pi(X,Y)}-(M+1)\pi(X,Y),
\]
an improvement on the bound of
$(2M+1)\sqrt{\pi(X,Y)}$ of Proposition 11 of \cite{HR19}.

For a simple example of the use of Theorem \ref{th:dist(X,Y)} where
Proposition 11 of \cite{HR19} does not apply: take $X$ to be a random
valuation that is neither bounded nor integrable (such as the 1-good
\textquotedblleft equal revenue" valuation in \cite{HN12}), and let $Y$ be a
$\delta$-perturbation of $X$ (such as moving every value of $X$ by at most
$\delta$ in the $\ell_{1}$-norm, for instance, to a nearby grid point). Thus
$\mathrm{W}(X,Y)\leq\delta$, and so if the revenue of $X$ is finite then
$\left\vert \mathrm{\text{\textsc{Rev}}}(X)-\,\text{\textsc{Rev}%
}(Y)\right\vert \leq2\sqrt{\delta\,\text{\textsc{Rev}}(X)}+\delta$ by
Corollary \ref{c:M} (and, if the revenue of $X$ is infinite then so is the
revenue of $Y$).

\subsection{Mechanism Rescaling\label{sus:rescaling}}

It remains to clarify the \textquotedblleft rescaling" argument. We start from
the trivial observation that changing the unit of the values of all goods by a
fixed factor (e.g., from dollars to cents) does not affect the revenue
(indeed, change the unit of all payments $s(\cdot)$ as well, while keeping the
allocations $q(\cdot)$ unchanged). Equivalently, rescaling the values of all
goods by a constant factor $\lambda>0$ rescales the revenue by that same
factor $\lambda$; i.e.,

\begin{proposition}
\label{p:rescale}For every random valuation $X$ and $\lambda>0,$%
\[
\Rev(\lambda X)=\lambda\,\Rev(X).
\]

\end{proposition}

For the formal proof we introduce the notion of \emph{mechanism rescaling},
which will be useful in the sequel: given a mechanism $\mu=(q,s)$ and
$\lambda>0,$ define the $\lambda$\emph{-rescaled} mechanism $\mu_{\lambda
}=(q_{\lambda},s_{\lambda})$ by $q_{\lambda}(\lambda x)=q(x)$ and $s_{\lambda
}(\lambda x)=\lambda s(x),$ i.e.,\footnote{In terms of ``pricing functions," if $p$ is
a pricing function of $\mu$ then $\lambda p$ is a pricing function of
$\mu_{\lambda}.$}%
\begin{equation}
q_{\lambda}(x):=q\left(  \frac{1}{\lambda}x\right)  \text{\ \ and\ \ }%
s_{\lambda}(x):=\lambda s\left(  \frac{1}{\lambda}x\right)
\label{eq:mu-lambda}%
\end{equation}
for every $x\in\mathbb{R}_{+}^{k}$. It is straightforward to verify that
$\mu_{\lambda}$ is IC and IR if and only if $\mu$ is IC and IR. The identity
$s_{\lambda}(\lambda x)=\lambda s(x)$ yields
\begin{equation}
R(\mu_{\lambda};\lambda X)=\lambda R(\mu;X), \label{eq:R-lambda}%
\end{equation}
i.e., the amount that the $\lambda$-rescaled mechanism $\mu_{\lambda}$
extracts from $\lambda X$ is exactly $\lambda$ times the amount that the
original mechanism $\mu$ extracts from $X.$ This proves Proposition
\ref{p:rescale}.

\bigskip

We can now provide some additional intuition for the proof of Proposition
\ref{p:rev-sqrt} above: we compare the revenue of a mechanism $\mu$ from $X$
with its revenue from $\lambda Y,$ which is the same as the revenue of the
$1/\lambda$-rescaled mechanism $\tilde{\mu}=\mu_{1/\lambda}$ from $Y$. As
discussed in the Introduction, the argument here is a sharpening and
simplification of the argument of \cite{HR19},\footnote{\cite{HR19} did not
use the fact that $\tilde{\mu}$ is just $\mu_{1/\lambda}$.} which is, in turn,
a simpler version of the argument of Nisan where the allocations were moved as
well (and so the mechanism $\tilde{\mu}$ is only defined implicitly).

\section{Revenue Continuity\label{s:continuity}}

An immediate consequence of the approximation result of Theorem
\ref{th:dist(X,Y)} is the continuity of the revenue with respect to the
Wasserstein distance, i.e., Theorem \ref{th:W-lim}:
\[
\text{if~~~}\mathrm{W}(X_{n},X)\rightarrow0\text{~~~then~~~}\text{\textsc{Rev}%
}(X_{n})\rightarrow\text{\textsc{Rev}}(X).
\]

\bigskip

\noindent\textbf{Remarks. }\emph{(a)} The convergence result holds whether
\textsc{Rev}$(X)$ is finite or infinite (in the latter case \textsc{Rev}%
$(X_{n})$ must be infinite for all but finitely many $n$).

\emph{(b) }The specific Wasserstein distance $\mathrm{W}$ that we have
considered uses the $\ell_{1}$-norm on $\mathbb{R}^{k}$ and then the $L_{1}%
$-norm on $\mathbb{R}^{k}$-valued functions (random valuations, i.e.,
$\mathbb{R}_{+}^{k}$-valued random variables). Using instead the $\ell_{p}$
and $L_{q}$ norms, for $1\leq p,q\leq\infty$---denote the resulting distance
$\mathrm{W}_{L_{q},\ell_{p}}$---can only yield stronger notions of
convergence. Indeed, the $\ell_{1}$ and $\ell_{p}$ norms are equivalent
($\left\Vert z\right\Vert _{p}\leq\left\Vert z\right\Vert _{1}\leq
k^{1-1/p}\left\Vert z\right\Vert _{p}$) and the $L_{1}$ norm is weaker than
the $L_{q}$ norm ($\left\Vert Z\right\Vert _{L_{1}}\leq\left\Vert Z\right\Vert
_{L_{q}}$); thus, $\mathrm{W}\equiv\mathrm{W}_{L_{1},\ell_{1}}\rightarrow0$ if
and only if $\mathrm{W}_{L_{1},\ell_{p}}\rightarrow0,$ and $\mathrm{W}%
\equiv\mathrm{W}_{L_{1},\ell_{1}}\rightarrow0$ is implied by (but need not
imply) $\mathrm{W}_{L_{q},\ell_{p}}\rightarrow0$. Thus,  for any $1\leq p,q\leq
\infty$, we have%
\[
\text{if~~~}\mathrm{W}_{L_{q},\ell_{p}}(X_{n},X)\rightarrow0\text{~~~then~~~}%
\text{\textsc{Rev}}(X_{n})\rightarrow\text{\textsc{Rev}}(X).
\]

\emph{(c)} Convergence with respect to the Wasserstein distance,
$\mathrm{W}(X_{n},X)\rightarrow0$, always implies convergence in distribution
(by (\ref{eq:pi-le-sqrt(w)})), and, when the $X_{n}$ and $X$ are integrable,
it is equivalent to convergence in distribution together with the uniform
integrability of the sequence\footnote{The sequence $X_{n}$ is \emph{uniformly
integrable} if for every $\varepsilon>0$ there is $M<\infty$ such that
$\mathbb{E}\left[  \left\Vert X_{n}\right\Vert \mathbf{1}_{\left\Vert
X_{n}\right\Vert >M}\right]  <\varepsilon$ for all $n$ (it does not matter
which norm $\left\Vert \cdot\right\Vert $ is used on $\mathbb{R}^{k}$; cf.
Remark (b) above).}$^{,}$\footnote{The equivalence is seen by using
Skorokhod's representation theorem \cite{Sko56}, \cite{Billingsley99} (Theorem
6.7).} $X_{n}$. Thus Theorem \ref{th:W-lim} implies Theorem 12 of \cite{HR19}.
However, for convergence in the Wasserstein distance the $X_{n}$ need not be
uniformly integrable, in fact not even integrable (for a simple example, take
$X_{n}=X+Z_{n}$ where $X$ is not integrable and $\mathbb{E}\left[  \left\Vert
Z_{n}\right\Vert \right]  \rightarrow0$).

\section{Approximation by Discount\label{s:rescale}}

The result of Theorem \ref{th:dist(X,Y)} can be restated as follows: given a
random valuation $X$ and a mechanism $\mu$ that extracts $r$ from $X$ (i.e.,
$r=R(\mu;X)$), for every random valuation $Y$ at Wasserstein distance
$\delta=\mathrm{W}(X,Y)$ from $X$ there is a mechanism $\mu^{\prime}$ that
extracts from $Y$ at least $r^{\prime}$ (i.e., $R(\mu^{\prime};Y)\geq r^{\prime}$), where $\sqrt{r^{\prime}}=\sqrt
{r}-\sqrt{\delta}$ (indeed,
$R(\mu;X)=r$ implies \textsc{Rev}$(X)\geq r,$ and so, for $\mu^{\prime}$ that
is optimal for $Y$---assume for now that it exists---we have $\sqrt{R(\mu^{\prime}%
;Y)}=\sqrt{\,\text{\textsc{Rev}}(Y)}\geq\sqrt{\,\text{\textsc{Rev}}(X)}%
-\sqrt{\mathrm{W}(X,Y)}\geq\sqrt{r}-\sqrt{\delta}=\sqrt{r^{\prime}}$). The
mechanism $\mu^{\prime}$ here depends on $Y.$ However, the proof in Section
\ref{s:approx} suggests that $\mu^{\prime}$ may be taken to be just a rescaling of
$\mu.$ This is formalized in Theorem \ref{th:C}, which says in particular the following: to get a loss of no more than
$\varepsilon$ in the revenue extracted (i.e., $r^{\prime}\geq r-\varepsilon$),
one takes as $\mu^{\prime}$ the $\eta$\emph{-discounting} $\mu_{1-\eta}$ of
$\mu$ for some discount rate $0<\eta<1$; moreover, this same mechanism $\mu
^{\prime}=\mu_{1-\eta}$ applies to \emph{all} $Y$ that are at a Wasserstein
distance $\delta$ or less from\footnote{For this specific use of Theorem
\ref{th:C}, where $X$ and $\mu$ are given, use $M=R(\mu;X)$.} $X$. We
recall the elementary \textquotedblleft change of units" transformation that
yields the $\eta$-discounted mechanism $\mu_{1-\eta}=(q_{1-\eta},s_{1-\eta})$
from the mechanism $\mu=(q,s)$:%
\begin{equation}
q_{1-\eta}(x)=q\left(  \frac{1}{1-\eta}x\right)  \text{\ \ and\ \ }s_{1-\eta
}(x)=(1-\eta)s\left(  \frac{1}{1-\eta}x\right)  .\label{eq:mu_(1-eta)}%
\end{equation}

We now prove Theorem \ref{th:C}.

\bigskip

\begin{proof}
[Proof of Theorem \ref{th:C}]Let $X$ and $Y$ be two random valuations, let
$\mu$ be a mechanism, and assume that $\mathrm{W}(X,Y)\leq\delta$ and
$R(\mu;X)\leq M$. For every $\lambda>1,$ dividing inequality (\ref{eq:R-R}) by
$\lambda$ yields%
\[
\frac{1}{\lambda}R(\mu;\lambda Y)\geq\frac{1}{\lambda}R(\mu;X)-\frac
{1}{\lambda-1}\mathrm{W}(X,Y).
\]
The lefthand side is $R(\mu_{1/\lambda};Y)$ (see (\ref{eq:mu-lambda})), and
so, for $1/\lambda=1-\eta,$ we have%
\begin{align*}
R(\mu_{1-\eta};Y)  & \geq(1-\eta)R(\mu;X)-\frac{1-\eta}{\eta}\mathrm{W}%
(X,Y)\\
& =R(\mu;X)-\eta R(\mu;X)-\frac{1-\eta}{\eta}\mathrm{W}(X,Y)\\
& \geq R(\mu;X)-\left(  \eta M+\frac{1-\eta}{\eta}\delta\right)  .
\end{align*}
When\footnote{At this point one can just substitute $\delta=\varepsilon
^{2}/(4M)$ and $\eta=\varepsilon/(2M)$ and complete the proof. The minimizing
argument that we use shows that this choice yields essentially the highest $\delta$
(see also Remark (d) after the proof).} $\eta=\sqrt{\delta/M}<1$ the
expression $\eta M+(1-\eta)\delta/\eta$ is minimal and equals $2\sqrt{M\delta
}-\delta;$ thus,
\begin{equation}
R(\mu_{1-\eta};Y)\geq R(\mu;X)-\left(  2\sqrt{M\delta}-\delta\right)
.\label{eq:1-eta}%
\end{equation}
Taking $\delta=\varepsilon^{2}/(4M)$ together with $\eta=\sqrt{\delta
/M}=\varepsilon/(2M)$ (which is indeed $<1$) yields inequality (\ref{eq:eps}),
and the proof is complete.

For the moreover statement, if $\mu$ is optimal for $X$ then, again for
$\eta=\sqrt{\delta/M},$ we have%
\begin{align*}
R(\mu_{1-\eta};Y) &  \geq R(\mu;X)-\left(  2\sqrt{M\delta}-\delta\right)  \\
&  =\text{\textsc{Rev}}(X)-2\sqrt{M\delta}+\delta\\
&  \geq\left(  \text{\textsc{Rev}}(Y)-2\sqrt{M\delta}-\delta\right)
-2\sqrt{M\delta}+\delta\\
&  =\text{\textsc{Rev}}(Y)-4\sqrt{M\delta}%
\end{align*}
(for the third line we used (\ref{eq:sqrt(d)}) in Corollary \ref{c:M}); now
$4\sqrt{M\delta}=2\varepsilon$ for $\delta=\varepsilon^{2}/(4M).$
\end{proof}

\bigskip

\noindent\textbf{Remarks. }\emph{(a)} The additive approximation in
(\ref{eq:eps}) needs, in order to determine the appropriate $\delta$ and
$\eta,$ an upper bound $M$ on the revenue that $\mu$ extracts from $X.$ If
$R(\mu;X)$ has been computed, one may well use it as $M.$ Otherwise, such
bounds may be $M=\,$\textsc{Rev}$(X)$ or $M=\mathbb{E}\left[  \left\Vert
X\right\Vert _{1}\right]  .$ If we have instead a bound $M_{Y}$ on the revenue
of the other valuation $Y,\,$i.e., $M_{Y}\geq\,$\textsc{Rev}$(Y),$ for
instance $M_{Y}=\mathbb{E}\left[  \left\Vert Y\right\Vert _{1}\right]  ,$ then
we can take $M=\left(  \sqrt{M_{Y}}+\sqrt{\varepsilon/2}\right)  ^{2},$
because $R(\mu;X)\leq\,$\textsc{Rev}$(X)\leq\left(  \sqrt{\text{\textsc{Rev}%
}(Y)}+\sqrt{\delta}\right)  ^{2}<\left(  \sqrt{M_{Y}}+\sqrt{\varepsilon
/2}\right)  ^{2}$ by Theorem \ref{th:dist(X,Y)} and $\delta = \varepsilon^2/(4M)<\varepsilon/2$ (since $\varepsilon < 2M$). This is relevant in
applications (such as sampling, as discussed in the Introduction) where we
want to determine the required bound $\delta$ on the Wasserstein distance
before we have access to $X.$

(\emph{b)} If a multiplicative approximation suffices then we can dispose of
the bound $M.$ Indeed, in the above proof we obtained%
\[
R(\mu_{1-\eta};Y)\geq(1-\eta)R(\mu;X)-\frac{1-\eta}{\eta}\delta,
\]
and so, for $\eta=\varepsilon$ and $\delta=\varepsilon^{2}/(1-\varepsilon)$
(assume that $\varepsilon<1$),%
\[
R(\mu_{1-\eta};Y)\geq(1-\varepsilon)R(\mu;X)-\varepsilon.
\]

\emph{(c)} If $\mu$ is approximately optimal for $X,$ specifically, $\rho
$-optimal for some $\rho>0$ (i.e., $R(\mu;X)\geq\,$\textsc{Rev}$(X)-\rho$),
then a simple adaptation of the above argument shows that $\mu_{1-\eta}$ is
guaranteed to be $(\rho+2\varepsilon)$-optimal for any $Y$ that is $\delta
$-close to $X$ in the Wasserstein distance.\footnote{For simplicity, use
$M\geq\,$\textsc{Rev}$(X)$ here.}

\emph{(d) }While one may always replace $\delta$ with a smaller $\delta
^{\prime}<\delta$, the discounting parameter $\eta$ needs to be
\textquotedblleft right," i.e., small but not too small (because in the
expression $\eta M+(1-\eta)\delta/\eta$ both $\eta$ and $1/\eta$ appear;
recall that the choice $\eta=\sqrt{\delta/M}$ was best). Also, any choices
of $\delta$ and $\eta$ must satisfy $\eta=O(\varepsilon/M)$ and $\delta
=O(\varepsilon^{2}/M)\ $as $\varepsilon\rightarrow0$ (because $\eta M$ and
$(1-\eta)\delta/\eta$ must both be $\leq\varepsilon$).

\emph{(e)} The following is an alternative formulation of the statement at the
beginning of this section: for any two random valuations $X$ and $Y$ and
mechanism $\mu$ such that\footnote{When $\mathrm{W}(X,Y)\geq R(\mu;X)$ the
inequality (\ref{eq:sqrt(mu-eta)}) is trivially satisfied by any $\mu^{\prime
}$ in place of $\mu_{1-\eta}$. } $\mathrm{W}(X,Y)<R(\mu;X)<\infty,$ the $\eta
$-discounting $\mu_{1-\eta}$ of $\mu$ with
\[
\eta=\sqrt{\frac{\mathrm{W}(X,Y)}{R(\mu;X)}}%
\]
satisfies%
\begin{equation}
\sqrt{R(\mu_{1-\eta};Y)}\geq\sqrt{R(\mu;X)}-\sqrt{\mathrm{W}(X,Y)}%
.\label{eq:sqrt(mu-eta)}%
\end{equation}
Indeed, (\ref{eq:1-eta}), with $M=R(\mu;X),$ $\delta=\mathrm{W}(X,Y)$, and
$\eta=\sqrt{\delta/M}$, yields $R(\mu_{1-\eta};Y)\geq M-2\sqrt{M\delta}%
+\delta=(\sqrt{M}-\sqrt{\delta})^{2},$ and hence (\ref{eq:sqrt(mu-eta)}) by
taking the square root (recall that we have assumed $M>\delta$).

\section{Classes of Mechanisms\label{s:classes of mechanisms}}

A class of mechanisms $\mathcal{N}$ is closed under \emph{mechanism}
\emph{rescaling} if for every mechanism $\mu$ belonging to the class
$\mathcal{N}$, all its rescaled mechanisms $\mu_{\lambda}$ for $\lambda>0$ also
belong to $\mathcal{N}$. Since only rescaled mechanisms were used to obtain
the results of the previous sections, they extend \emph{mutatis mutandis} to
any such $\mathcal{N}.$ This includes all standard classes of mechanisms, such
as separate (\textsc{SRev}), bundled (\textsc{BRev}), deterministic
(\textsc{DRev}), finite menu (\textsc{MRev}$_{[m]}$), monotonic
(\textsc{MonRev}), allocation monotonic (\textsc{AMonRev}), submodular
(\textsc{SubModRev}), supermodular (\textsc{SupModRev}), and buy-many.\footnote{We
do not know any interesting class of mechanisms that does not satisfy this
condition.} Formally,

\begin{theorem}
Let $\mathcal{N}$ be a class of mechanisms that is closed under mechanism rescaling.

(i) Let $X$ and $Y$ be $k$-good random valuations such that the Wasserstein
distance $\mathrm{W}(X,Y)$ is finite. Then either $X$ and $Y$ both have
infinite $\mathcal{N}$-revenue, or both have finite $\mathcal{N}$-revenue,
and
\[
\left\vert \sqrt{\mathcal{N}\text{-}\Rev(X)}-\sqrt{\mathcal{N}\text{-}\Rev(Y)}\right\vert \leq\sqrt{\mathrm{W}(X,Y)}.
\]
(ii) Let $X$ and $\{X_{n}\}_{n\geq1}$ be $k$-good random valuations. If
$X_{n}$ converges to $X$ in the Wasserstein distance, i.e., $\mathrm{W}%
(X_{n},X)\rightarrow0$, then%
\[
\lim_{n\rightarrow\infty}\mathcal{N}\text{-}\Rev(X_{n})=\mathcal{N}%
\text{-}\Rev(X).
\]

\end{theorem}

\bigskip

Finally, we observe that, since discounted mechanisms are rescaled mechanisms,
the mechanism $\mu_{1-\eta}$ of Theorem \ref{th:C} is in the same class
$\mathcal{N}$ as $\mu$ (for any $\mathcal{N}$ that is closed under mechanism rescaling).

\section{Other Allocation Spaces\label{s:other Gamma}}

Let $\Gamma\subset\mathbb{R}^{k}_+$ be a compact space of allocations. A
mechanism $\mu=(q,s)$ is a $\Gamma$-\emph{mechanism}\textbf{\ }if\textbf{\ }%
$q(x)\in\Gamma$ for every $x\in\mathbb{R}_{+}^{k}.$ We denote by
\textsc{Rev}$_{\Gamma}(X)$ the maximal revenue that can be achieved by IC and
IR $\Gamma$-mechanisms. Examples:

\begin{enumerate}
\item \emph{additive demand} of $k$ goods (the model of the previous sections): $\Gamma_{1}=[0,1]^{k}$ (general mechanisms), and $\Gamma_{1}^D=\{0,1\}^{k}$ (deterministic mechanisms); 

\item \emph{unit demand} of $k$ goods, or \emph{implementation} of $k$
outcomes with a null (no participation) option\footnote{In case 2 one may choose
one of the $k$ options or \textquotedblleft nothing"; in case 3 one
\emph{must} choose one of the $k$ options.}: $\Gamma_{2}=\{g\in\lbrack
0,1]^{k}:\sum_{i}g_{i}\leq1\}$ (general mechanisms) and $\Gamma_2^D = \{\mathbf{0},e^1,...,e^k\}$ (deterministic mechanisms), where $\mathbf{0}$ denotes the all-$0$ vector, and $e^i$ the $i$th unit vector;

\item \emph{implementation} of $k$ outcomes (without a null option):
$\Gamma_{3}=\{g\in\lbrack0,1]^{k}:\sum_{i}g_{i}=1\}$ (general mechanisms) and
$\Gamma_3^D = \{e^1,...,e^k\}$ (deterministic mechanisms).

\end{enumerate}

Each such set $\Gamma$ defines a \emph{dual }$\Gamma$-\emph{seminorm}%
\footnote{This says that $\rho_{\Gamma}(\alpha z)=\left\vert \alpha\right\vert
\rho_{\Gamma}(z)$ for any scalar $\alpha$ and $\rho_{\Gamma}(z+z^{\prime}%
)\leq\rho_{\Gamma}(z)+\rho_{\Gamma}(z^{\prime}).$ To be a norm one needs in
addition that $\rho_{\Gamma}(z)=0$ (if and) only if $z=\mathbf{0}$. This holds
when $B_{\Gamma}$ is fully dimensional, as is the case for $\Gamma_{1}$ and
$\Gamma_{2}$, but not for $\Gamma_{3}$ (where $B_{\Gamma_{3}}\subseteq
\{h:\sum_{i}h_{i}=0\}$ and $\rho_{\Gamma_{3}}(1,...,1)=0$).} on $\mathbb{R}%
^{k}$ by taking $B_{\Gamma}:=\rm{conv}(\Gamma-\Gamma)$ as the unit ball:%
\[
\rho_{\Gamma}(z):=\max_{h\in B_{\Gamma}}h\cdot z=\max_{g\in\Gamma}g\cdot
z-\min_{g\in\Gamma}g\cdot z.
\]
In the examples above we get (in each case, in the general setup as well as in the deterministic setup):

\begin{enumerate}
\item $\rho_{\Gamma_{1}}(z)=\sum_{i}\left\vert z_{i}\right\vert =\left\Vert
z\right\Vert _{1}$ (here $B_{\Gamma_{1}}=B_{\Gamma_{4}}$ is the $\ell_{\infty}$ unit ball,
and so $\rho_{\Gamma_{1}}$ is the $\ell_{1}$-norm);

\item $\rho_{\Gamma_{2}}(z)=\max\{z_1,...,z_k,0\}-\min\{z_1,...,z_k,0\}$;

\item $\rho_{\Gamma_{3}}(z)=\max_{i}z_{i}-\min_{i}z_{i}$ (the
\textquotedblleft\emph{spread"} of $z$).
\end{enumerate}

All the results of this section hold \emph{mutatis mutandis} for $\Gamma
$-mechanisms, with the $\ell_{1}$-norm $\left\Vert \cdot\right\Vert _{1}$
replaced by the dual $\Gamma$-seminorm $\rho_{\Gamma}(\cdot).$ Indeed, when
$q(x)$ and $q(\lambda y)$ belong to $\Gamma$ we have%
\[
(q(x)-q(\lambda y))\cdot(x-y)\leq\rho_{\Gamma}(x-y),
\]
and so the result of Lemma \ref{l:lambda} becomes%
\[
s(x)\leq s(\lambda y)+\frac{\lambda}{\lambda-1}\rho_{\Gamma}(x-y).
\]
Thus Theorem \ref{th:dist(X,Y)} generalizes to%
\[
\left\vert \sqrt{\text{\textsc{Rev}}_{\Gamma}(X)}-\sqrt{\text{\textsc{Rev}%
}_{\Gamma}(Y)}\right\vert \leq\sqrt{\mathrm{W}_{\Gamma}(X,Y)},%
\]
where%
\begin{align*}
\mathrm{W}_{\Gamma}(X,Y):=  &  \inf\,\{\,\mathbb{E}\left[  \rho_{\Gamma
}(X^{\prime}-Y^{\prime})\right]  \,:\,X^{\prime},Y^{\prime}\text{ are defined
on the same space,}\\
&  ~\;\;\;\hspace{1.3in}X^{\prime}\overset{\mathcal{D}}{\sim}X\text{ and
}Y^{\prime}\overset{\mathcal{D}}{\sim}Y\}.
\end{align*}

In case 1, the additive-demand case, we have $\mathrm{W}_\Gamma=\mathrm{W}$, which shows in particular that the deterministic revenue $\normalfont\textsc{DRev}$ satisfies the same inequality as the revenue $\Rev$:
\[
\left\vert \sqrt{\text{\normalfont\textsc{DRev}}(X)}-\sqrt{\text{\normalfont\textsc{DRev}%
}(Y)}\right\vert \leq\sqrt{\mathrm{W}(X,Y)}.%
\]
In the other two cases we have $\mathrm{W}_{\Gamma}\leq2\mathrm{W}%
_{L_{1},\ell_{\infty}}\leq2\mathrm{W}_{L_{1},\ell_{1}}\equiv2\mathrm{W}$
(because $\rho_{\Gamma}(z)\leq2\left\Vert z\right\Vert _{\infty}%
\leq2\left\Vert z\right\Vert _{1}$) (again, in the general as well as in the deterministic setups). For a general $\Gamma$ we get
$\mathrm{W}_{\Gamma}\leq2\gamma\mathrm{W}$ where $\gamma:=\max_{g\in\Gamma
}\left\Vert g\right\Vert _{\infty}$ (which is finite because $\Gamma$ is a
compact set).

Therefore the general convergence result is%
\[
\text{if~~~}\mathrm{W}_{\Gamma}(X_{n},X)\rightarrow0\text{~~~then~~~}%
\text{\textsc{Rev}}_{\Gamma}(X_{n})\rightarrow\text{\textsc{Rev}}_{\Gamma}(X)
\]
(as seen above, $\mathrm{W}_{\Gamma}(X_{n},X)\rightarrow0$ is implied
by $\mathrm{W}(X_{n},X)\rightarrow0$).

Since mechanism discounting does not affect the allocations of goods (it just
transforms the mapping from valuations to allocations), the result of Theorem
\ref{th:C} holds for $\Gamma$-mechanisms as well.

Therefore the results for the two setups discussed in the
Introduction---approximation and imprecision, and sampling and
learning---continue to hold in these more general models.

\bibliographystyle{alpha}
\bibliography{bib.bib}

\end{document}